\documentclass{svproc}
\usepackage{braket,amsmath,amssymb,todonotes,listings}
\usepackage[altpo]{backnaur}
\usetikzlibrary{fit,calc,arrows,positioning, graphs,graphs.standard}

\newcommand{\makeclause}[3]%
  {#1\ \leftarrow\ #2\ \&\ \ldots\ \&\ #3}

\newcommand{\dispform}[2]%
      {
       \begin{equation}
       #2
       \end{equation}
      }
\newcounter{equn}[section]
\newcommand{\eqnsect}[1]%
      {\refstepcounter{equn}
       
       \begin{equation}
       #1
       \end{equation}
      }

\newcommand{\ARGS}%
    {\mbox{lub}\ \pi_{1}[V],\ \ldots,\ \mbox{lub}\ \pi_{n}[V]}

\newtheorem{NS-Experiment}{Experiment}

\newtheorem{NS-Fact}{Fact}

\newtheorem{Example}{Example}
\newtheorem{NS-Example}{Example}

\newtheorem{Definition}{Definition}
\newtheorem{NS-Definition}{Definition}

\newtheorem{Proposition}{Proposition}
\newtheorem{NS-Proposition}{Proposition}

\newtheorem{NS-Homework}{Homework}

\newtheorem{NS-Theorem}{Theorem}

\newtheorem{NS-Observation}{Observation}

\newtheorem{NS-Lemma}{Lemma}

\newtheorem{Remark}{Remark}
\newtheorem{NS-Remark}{Remark}

\newtheorem{NS-Corollary}{Corollary}

\newtheorem{NS-Claim}{Claim}

\newtheorem{NS-claim}{claim}

\newtheorem{NS-Directive}{Directive}

\newtheorem{NS-Problem}{Problem}

\newtheorem{NS-Solution}{Solution}

\newtheorem{Notation}{Notation}
\newtheorem{NS-Notation}{Notation}

\newtheorem{NS-Note}{Note}

\newtheorem{NS-Question}{Question}

\newtheorem{NS-Answer}{Answer}

\newtheorem{NS-Counterexample}{Counterexample}

\newcommand{\st}\backepsilon
\newcommand{\suchthat}{\mathrel{\ooalign{$\ni$\cr\kern-1pt$-$\kern-6.5pt$-$}}}

\input{jackmac}

\usepackage{url}

\begin{document}
	\mainmatter              
	\title{Quantum State Diffusion on a Graph}
	\titlerunning{Quantum State Diffusion on a Graph}  
	%
	\author{John C Vining III\inst{1} \and Howard A Blair\inst{1}
	}
	%
	%
	\tocauthor{Jack Vining, Howard A Blair}
	\institute{Department of Electrical Engineering and Computer Science, Syracuse University\\
		\email{jcvining@syr.edu}
	}
	
	\maketitle              
	
	\begin{abstract}
		Quantum walks have frequently envisioned the behavior of a quantum state traversing a classically defined, generally finite, graph structure. While this approach has already generated significant results, it imposes a strong assumption: all nodes where the walker is not positioned are quiescent. This paper will examine some mathematical structures that underlie state diffusion on arbitrary graphs, that is the circulation of states within a graph. We will seek to frame the multi-walker problem as a finite quantum cellular automaton. Every vertex holds a walker at all times. The walkers will never collide and at each time step their positions update non-deterministically by a quantum swap of walkers at opposite ends of a randomly chosen edge. The update is accomplished by a unitary transformation of the position of a walker to a superposition of all such possible swaps and then performing a quantum measurement on the superposition of possible swaps. This behavior generates strong entanglement between vertex states which provides a path toward developing local actions producing diffusion throughout the graph without depending on the specific structure of the graph through blind computation.
	\end{abstract}
	
	\section{Introduction}
	Graph walks can be seen as the behavior of a point wandering through a graph's vertex set, constrained to only transitions between vertices which share an edge.  A walker then can also be seen as modeling the diffusion of states through the graph.  As an illustration of the potential behavior of a quantum walking system, visualize a palette with several distinct colors where wells are vertices and colors are states.  Over time the wells of paint are mixed with their neighbors, resulting in a muddling of the colors in each.  As a quantum system, however, when we directly observe the contents of a specific well, we collapse it into one of the original, pure colors.  This collapse also effects the wells not measured, as the presence of a color in one well excludes the possibility of it's presence elsewhere, resulting in the zeroing out its probability everywhere else on the palette.  
	\subsection{Graph Walkers}
	The application of classical graph walks has been well established as a highly effective method of solving certain classes of problems.  Its natural extension to quantum graph walks which have demonstrated superior performance for certain problems\cite{kempe_quantum_2003}.  There are also classes of graph structures that quantum walks progress through exponentially faster than classical walks\cite{keating_localization_2007}.  
	\subsection{Traditional Quantum Graph Walker}
	Most implementations of quantum walkers are implemented by storing the position of a walker in a quantum register, then evolving this state based upon a known state transition system and a appropriate coin states\cite{brun_quantum_2003}.  Quantum walks which deal with the behavior of multiple walkers are less common, presumably due to the limits of current quantum devices and researchers focusing their work on more realizable implementations.  
	The work on these multi-walker models is primarily divided into two forms: non-interacting walkers and boson based systems.  The first, non-interacting states bypasses the very traits we wish to study in the future.  Boson based quantum states have been studied\cite{rohde_multi-walker_2011} and while they can allow for walkers interacting with one another, by their nature boson systems would allow multiple walkers to be physically present in the same region at the same time.  Measuring then could result in the destruction of walkers.
	
	\section{Alternative Quantum Graph Walker}
	The quantum walker model proposed here is implemented on an undirected graph $\Gamma = \left(V,E\right)$. The quantum system is composed of a quantum register for each vertex of $V$ and fixed quantum channels between each vertex's register and its neighbor's as defined in $E$.  It should be noted that we will impose an arbitrary direction on each edge as we construct the system to simplify the process as described.  By their nature a quantum channel\cite{van_meter_quantum_2014} is bidirectional as their behavior is expressible as a unitary, and therefore any action taken in one direction can be taken in the other.
	\subsection{State Labeling}
	To clarify the labeling style of our Dirac notation we specify the type and physical location of quantum states.  Let $\Gamma= \left(V,E\right)$ be an oriented graph, that is a directed graph where the edge between two vertices is never bi-directional.
	We lay out a brief grammar of our labels:
	\begin{bnf*}
		\bnfprod{vertex}
		{\bnftd{label of a vertex of the graph}}\\
		\bnfprod{edge}
		{\bnfpn{vertex} \bnfts{,} \bnfpn{vertex}}\\
		\bnfprod{number}{
			\bnfts{1} \bnfor \bnfts{2}
		}\\
		\bnfprod{neighborhood-set}
		{\bnfpn{vertex} \bnfts{,$\square$}}\\
		\bnfprod{vertex-state-label}{
			\bnfpn{vertex}
		}\\	\bnfprod{vertex-proxy-label}{
			\bnfpn{vertex} \bnfts{'}
		}\\
		\bnfprod{edge-flag-state-label}{
			\bnfpn{edge} \bnfts{,} \bnfpn{number}}\\
		\bnfprod{edge-state-label}{
			\bnfpn{edge} \bnfts{$,\square$}}\\
		\bnfprod{neighborhood-state-label}{
			\bnfpn{vertex} \bnfts{$,\square,\square$}\bnfor
			\bnfpn{vertex} \bnfts{$,\square,$} \bnfpn{number}}\\
	\end{bnf*}
	where all $\quoted{*-\text{label}}$ are valid quantum component labels.
	Thus for $\forall v\!\in V $ a $k$-level qudit is labeled $\quoted{v}$ and 
	$\forall \left(u,v\right)\! \in\! E$, a quantum register labeled $\quoted{u,v,\square}$ is composed of $2$ qubits labeled  
	$\lbrace \quoted{v,u,1},\quoted{v,u,2}\rbrace$.
	A vertex $v \in V$ then has a quantum register composed of two qudits labeled $\quoted{v}$ and $\quoted{v'}$ and a $2 \deg{v}$ qubit register labeled $\quoted{v,\square,\square}$ The vertex $\quoted{v}$ then has a quantum register structure:
	%
	
	\begin{tikzpicture}
		\node[rectangle,draw,anchor=east,minimum height=19pt](state){$\ket{\psi}_{v}$};
		\node[rectangle,draw,anchor=north,minimum height=19pt](stateprox) at (state.south){$\ket{0}_{v'}$};
		\node[rectangle,draw,anchor=west](outState1) at (state.east) {$\bigTensor_{u\in \neigh{v}} \ket{\psi_{u,1}}_{v,u,1}$};
		\node[rectangle,draw,anchor=north](outState2) at (outState1.south) {$\bigTensor_{u\in \neigh{v}} \ket{\psi_{u,2}}_{v,u,2}$};
	\end{tikzpicture}
	
	Additionally we use alternative labeling notation to partition the quantum register in different ways to facilitate a clearer understanding of their behavior.
	That is 
	\[\ket{\psi_1}_{v,\square,1} = \bigTensor_{u \in \neigh{v}}{\ket{\psi_{u,1}}_{v,u,1}}
	\hspace{5em} \ket{\psi_2}_{v,\square,2} = \bigTensor_{u \in \neigh{v}}{\ket{\psi_{u,2}}_{v,u,2}}\]
	and
	\[\ket{\psi}_{u,v,\square} = \ket{\psi_1}_{u,v,1}\ket{\psi_2}_{u,v,2}\]
	Thus we can also view the vertex quantum register as
	\begin{equation*}
		\vcenter{\hbox{\begin{tikzpicture}
					\node[rectangle,draw,anchor=east,minimum height=19pt](state){$\ket{\psi}_{v}$};
					\node[rectangle,draw,anchor=north,minimum height=19pt](stateprox) at (state.south){$\ket{0}_{v'}$};
					\node[rectangle,draw,anchor=west](outState1) at (state.east) {$\ket{\psi_{1}}_{v,\square,1}$};
					\node[rectangle,draw,anchor=north](outState2) at (outState1.south) {$\ket{\psi_{2}}_{v,\square,2}$};
		\end{tikzpicture}}}\hspace{2em}=\hspace{2em}
		\vcenter{\hbox{\begin{tikzpicture}
					\node[rectangle,draw,anchor=east,minimum height=19pt](state){$\ket{\psi}_{v}$};
					\node[rectangle,draw,anchor=north,minimum height=19pt](stateprox) at (state.south){$\ket{0}_{v'}$};
					\node[rectangle,draw,anchor=north west,minimum height=38pt](outState1) at (state.north east) {$\bigTensor_{u\in \neigh{v}} \ket{\psi_{u}}_{v,u,\square}$};
		\end{tikzpicture}}}
	\end{equation*}
	This notation allows us to reference slices of the quantum register in ways that work toward illuminating their role and behavior.
	\subsection{State Preparation}
	We now describe how coin state space\cite{brun_quantum_2003} for the vertex $v$ is prepared.  The quantum space labeled $\quoted{v,\square,1}$
	is placed in a $\deg{v}$-W-State.  
	\begin{Definition}[$n$-W-State\cite{du_three_2000}]
		A n-W-State is a $n$ qubit quantum pure state such that when the quantum register is measured in the standard basis, exactly 1 qubit will measure 1 while all others measure 0.  
	\end{Definition}
	\begin{Example}
		The 3-W-State is
		\[\frac{1}{\sqrt{3}} \left( \ket{001} + \ket{010} +\ket{100} \right)\]
	\end{Example}
	\begin{Definition}[Bell Paired]
		Two quantum components $\quoted{u}$ and $\quoted{v}$ are bell paired if whenever measured in the standard basis, their states are identical, that is if $\quoted{u}$ is in the state: 
		\begin{equation}
			\ket{\psi}_u = \sum_{i \in X}{\alpha_{i} \ket{i}_{u}}
		\end{equation}
		then the composite quantum system of $\quoted{u}$ and $\quoted{v}$ is:
		\begin{equation}
			\bigTensor_{i \in X}{\alpha_{i} \ket{i}_{u} \ket{i}_{v}}
		\end{equation}
	\end{Definition}
	The quantum register labeled $\quoted{v,\square,1}$
	is then bell paired with 
	$\quoted{v,\square,2}$
	in the state:\[
	\sum_{t \in \neigh{v}}{
		\frac{1}{\sqrt{\deg{v}}}
		\bigTensor_{u \in \neigh{v}}{
			\ket{\delta_{t,u}}_{v,u,1}
			\ket{\delta_{t,u}}_{v,u,2}
		}
	}
	\]
	where $\delta_{t,u}$ is the Kronecker delta on the labels $\quoted{t}$ and $\quoted{u}$.
	\subsection{Transformations}
	We then have the full graph in the state
	\[
	\bigTensor_{v \in V}{
		\ket{\psi_v}_{v} \otimes
		\sum_{t \in \neigh{v}}{
			\frac{1}{\sqrt{\deg{v}}}
			\bigTensor_{u \in \neigh{v}}{
				\ket{\delta_{t,u}}_{v,u,1}
				\ket{\delta_{t,u}}_{v,u,2}
			}
		}
	}
	\]
	Every vertex $v \in V$ then swaps its 
	$\left(v,u,1\right)$ 
	state with its neighbor $u \in \outNeigh{v}$'s 
	$\left(u,v,2\right)$
	thus then partial trace\cite[p.~105]{nielsen_quantum_2011} of the state of the graph at vertex v is:
	\[
	\ket{\psi_v}_{v} \otimes
	\sum_{t \in \neigh{v}}{
		\frac{1}{\sqrt{\deg{v}}}
		\bigTensor_{u \in \neigh{v}}{
			\left( \sum_{j \in \neigh{u}}{
				\frac{1}{\sqrt{\deg{v}}}
				\ket{\delta_{j,v}}_{v,u,1}}\right)
			\ket{\delta_{t,u}}_{v,u,2}
		}
	}
	\]
	or
	\[
	\ket{\psi_v}_{v} \otimes
	\sum_{t \in \neigh{v}}{
		\bigTensor_{u \in \neigh{v}}{
			\sum_{j \in \neigh{u}}{
				\frac{1}{\sqrt{\deg{u}}\sqrt{\deg{v}}}
				\ket{\delta_{j,v}}_{v,u,1}
				\ket{\delta_{t,u}}_{v,u,2}
			}
		}
	}
	\]
	thus $\ket{11}_{v,u,\square}$ will be measured only when $j=v \wedge t=u$ with probability 
	$\frac{1}{\deg{u}\deg{v}}$
	
	As we are assuming nodes do not have proximity to one another these swaps are implemented via bidirectional quantum teleportation.\cite{hassanpour_bidirectional_2015}
	Lastly we include in each node $\quoted{v}$ a qudit of identical dimension to the node's state labeled $\quoted{v`}$.  This state will act as a proxy for the following step.  
	When evaluating edge $v,u$, node $v$ will first perform a controlled swap on $\quoted{v}$ and $\quoted{v'}$, while $u$ will similarly perform a controlled swap on $quoted{u}$ and $\quoted{u'}$.  A standard bidirectional quantum teleportation is then performed between $v$ and $u$ to swap the components $\quoted{u'}$ and $\quoted{v'}$ and finally the controlled swaps are repeated by $u$ and $v$.  This process, while cumbersome is necessary to implement the controlled swap of two disparate quantum components that preserves the superposition due to the nature of its controls.
	\subsection{Leveraging the W-States}
	Consider a visualization of the W-States of the system. Let the system be described by 
	\[\bigTensor_{v \in V}\bigTensor_{u \in \neigh{v}}{\ket{\psi_v}_v \ket{\phi_{u,v}}_{v,u,1}
		\ket{\phi_{v,u}}_{v,u,2}}\]
	For all the configurations of $v \in V, u\in \neigh{v}, \phi_{v,u}$ consider highlighting edges such that for the directed edge $\left(a,b\right) \in E$, the half of the edge from $a$ is highlighted if $\phi_{a,b}=1 $ and half of the edge from $b$ is highlighted if $\phi_{b,a}=1$.
	
	\begin{Proposition}[Independent Edges]
		For all edges $\left(u,v \right) \in E$, the state of 
		$\ket{1}_{u,v,1}
		\ket{1}_{u,v,2} \implies \forall \left(j,k\right) \in E, $ incident to $\left(u,v\right)$, then 
		\[\ket{1}_{j,k,1}
		\ket{1}_{j,k,2} \implies
		\left(u,v\right) =   	 \left(j,k\right)
		\]
	\end{Proposition}
	\begin{proof}
		Let $\left( u,v \right) \in E$ and 
		$\ket{1}_{u,v,1}
		\ket{1}_{u,v,2} \implies \phi_{u,v} = \phi_{v,u}=1$.
		
		For the purpose of contradiction, suppose $\exists j \in \neigh{u}\backslash\{v\}$ such that 
		
		$\ket{1}_{j,u,1} \ket{1}_{j,u,2}$.  This implies that $\phi_{u,j} = \phi_{j,u} = 1$ thus $\phi_{u,v} = \phi_{u,j} = 1 \implies v = j$, a contradiction.
		
		Similarly suppose $\exists k \in \neigh{v}\backslash\{u\}$ such that 
		$\ket{1}_{k,v,1} \ket{1}_{k,v,2}$.  This implies that $\phi_{v,k} = \phi_{k,v} = 1$ thus $\phi_{v,u} = \phi_{v,k} = 1 \implies u = k$, a contradiction.
		
		Thus all edges incident to a selected edge are themselves not selected.
	\end{proof}
	\subsection{Edge selection probability distribution}
	For an edge $\left(u,v\right) \in E$, what is the expected probability that it is a member of a randomly selected independent edge set $S$.  Notationally let us write 
	$p_{u,v} = \rm{Prob}\!\left(\left(u,v\right) \in S \right)$.  We will also assume that the probabilities of non-adjacent edges are independent, that is 
	\[\forall (u,v),(j,k) \in E, 
	\rm{Prob}\!\left( (u,v)\!\in\! S \mid (j,k)\in S \right) =
	\begin{cases}
		p_{u,v} & \{u,v\} \cap \{j,k\} = \emptyset\\
		1		& \{u,v\} \cap \{j,k\} = \{u,v\}\\
		0 		& otherwise
	\end{cases}\]
	Additionally we require that
	\[
	p_{u,\!v} \leq
	\prod_{t \in \neigh{u}\backslash\{v\}}{\left(1-p_{t,\!u}\right)}
	\prod_{w \in \neigh{v}\backslash\{u\}}{\left(1-p_{v,\!w}\right)}
	\]
	as the edge 
	$\left(u,v\right) \in S \implies 
	\forall t \in \neigh{u}, \left(t,u\right) \notin S\land
	\forall w \in \neigh{v}, \left(v,w\right) \notin S
	$.
	It should be noted that these conditions are clearly satisfiable, as assigning a value of $p$ gives us
	\[\begin{aligned}
		\prod_{t \in \neigh{u}\backslash\{v\}}{\left(1-p_{t,\!u}\right)}
		\prod_{w \in \neigh{v}\backslash\{u\}}{\left(1-p_{v,\!w}\right)}
		=\left(1-p\right)^{\deg{u}-1}
		\left(1-p\right)^{\deg{v}-1}\\
		=\left(1-p\right)^{\deg{u}+\deg{v}-2}
		\geq \left(1-p\right)^{2\Delta}
	\end{aligned}\]
	Thus solving for $p = (1-p)^{2\Delta}$ gives a value $p$ such that 
	\[
	\begin{aligned}
		p = (1-p)^{2\Delta} = (1-p)^{\left(\deg{u} + a\right) + \left(\deg{v} +b\right)} = 
		(1-p)^{\deg{u} + \deg{v}}
		(1-p)^{a+b} \\
		\leq (1-p)^{\deg{u} + \deg{v} - 2}
		=\prod_{t \in \neigh{u}\backslash\{v\}}{\left(1-p\right)}
		\prod_{w \in \neigh{v}\backslash\{u\}}{\left(1-p\right)}
	\end{aligned}\]
	therefore
	\[p_{u,v}=
	\begin{cases}
		p& \left(u,v\right) \in E\\
		0&	otherwise
	\end{cases}\]
	satisfies the requirements.
	\subsection{Advantages}
	Current work on quantum walkers is focused primarily on their capacity to spread quickly into a superposition over a clearly defined graph structure.  Their evolution, however, is dictated by full knowledge of the structure of a graph and implementing an appropriate coin system to direct the state's evolution.  By reflecting the graph's structure in the set of local quantum components, with fixed quantum communication channels, we remove the burden of requiring full knowledge of the graph to determine the walker's local behavior.
	\section{Example}
	\subsection{Example functions}
	Let us prepare a set of quantum operators that help in preparing the W-States of the system $\{U_W^{v} \mid v \in V\}$ such that 
	\[U_W^{v} \bigTensor_{u \in \neigh{v}}{\ket{0}_{v,u,1} \ket{0}_{v,u,2}} =
	\frac{1}{\sqrt{\deg{v}}}\left(
	\sum_{j \in \neigh{v}}{\bigTensor_{u \in \neigh{v}}{\ket{\delta_{j,u}}_{v,u,1} \ket{\delta_{j,u}}_{v,u,2}}}
	\right)
	\]
	thus for example:
	\[
	\ket{0}_{A,B,1}\ket{0}_{A,C,1}
	\ket{0}_{A,B,2}\ket{0}_{A,C,2}\]
	is transformed into
	\[
	\frac{1}{\sqrt{2}}
	\left(
	\ket{1}_{A,B,1}\ket{0}_{A,C,1}
	\ket{1}_{A,B,2}\ket{0}_{A,C,2}
	+
	\ket{0}_{A,B,1}\ket{1}_{A,C,1}
	\ket{0}_{A,B,2}\ket{1}_{A,C,2}
	\right)\]
	As $U_W^v$ behavior on $\ket{0}^\otimes$ can be viewed as a change of basis, the unitary clearly exists.
	Initializing the system with the state
	\[
	\bigTensor_{v \in V}{
		\ket{v}_{v}
		\bigTensor_{u \in \neigh{v}}{\ket{0}_{v,u,1} \ket{0}_{v,u,2}}}
	\]
	where the quantum component labeled $\quoted{v}$ has an orthonormal basis labeled by $v \in V$. 
	\begin{Notation}[Composition]
		We define $\unitaryComposition_{v\in V}{U_v}$ as the composition of a set of unitary operators.  This generally requires a fixed order on the elements of $V$, as the composition of unitaries is typically not commutative.
	\end{Notation}
	\begin{Remark}
		In this paper when unitary composition notation is used, the unitaries involved have mutually exclusive controls such that any non-commutative unitaries never satisfy their controls simultaneously.  Thus in this paper, all unitaries composed in such a way are all commutative and no ordering is needed.
	\end{Remark}
	Applying $\unitaryComposition_{v \in V}{U_W^v}$ to the system places all coin states in their appropriate W-State superposition.  
	Applying the appropriate control swaps 
	\[\unitaryComposition_{v \in V}{
		\unitaryComposition_{u \in \outNeigh{v}}
		{\text{SWAP}_{\left(v,u,1\right),\left(u,v,1\right)}
		}
	}\]
	consolidates the coin states governing the edge $\left(u,v\right) \in E$ in the quantum register labeled $\quoted{u,v,\square}$.
	The vertices states are then exchanged by the doubly controlled swap gates 
	$\unitaryComposition_{v \in V}\unitaryComposition_{u \in \outNeigh{v}}\text{CCSWAP}_{
		\left(v,u,1\right),
		\left(v,u,2\right),
		\left(v\right),
		\left(u\right),}$
	
	Repeating the process results in the states of vertices permuting through the graph via quantum transformations limited to neighboring vertices.
	\subsection{Watrous Implementation}
	Watrous first introduced the principle of quantum cellular automata\cite{watrous_one-dimensional_1995}.  By partitioning a quantum register and implementing fixed circulation of relative block indices one can ensure there is a degree of cellular interaction.  Using this circulation as a basis, applying a fixed unitary to each block uniformly enables WQCA to perform a large class of calculations.  
	Consider the diffusion across a infinite line graph $\Gamma$ with vertices $\lbrace v_i \rbrace_{i \in \mathbb{Z}}$. 
	Conforming to the structure previously discussed then, we group each block of 3 qubits into their own isolated $C_3$, the quantum registry associated with the block of vertices $\quoted{3i},\quoted{3i+1},\quoted{3i+2}$ is Figure \ref{fig:WQCARegistryDiagram}. We diverge slightly from Watrous' original model by giving each block additional qubits to drive the diffusion as previously outlined. 
	\begin{figure}[!ht]
		\centering
		\caption{Quantum Registry of WQCA block}
		\label{fig:WQCARegistryDiagram}
		\begin{tikzpicture}
			\node[rectangle,draw,anchor=east](statea){$\ket{\psi}_{3i}$};
			\node[rectangle,draw,anchor=north](stateproxa) at (statea.south){$\ket{0}_{3i'}$};
			\node[rectangle,draw,anchor=west](outState1aa) at (statea.east) {$\ket{\psi_{3i+1,1}}_{3i,3i+1,1}$};
			\node[rectangle,draw,anchor=west](outState1ba) at (outState1aa.east) {$\ket{\psi_{3i+2,1}}_{3i,3i+2,1}$};
			\node[rectangle,draw,anchor=north](outState2aa) at (outState1aa.south) {$\ket{\psi_{3i+1,2}}_{3i,3i+1,2}$};
			\node[rectangle,draw,anchor=west](outState2ba) at (outState2aa.east) {$\ket{\psi_{3i+2,2}}_{3i,3i+2,2}$};
			\node[rectangle,draw,anchor=north east](stateb) at (outState2aa.south west){$\ket{\psi}_{3i+1}$};
			\node[rectangle,draw,anchor=north](stateproxb) at (stateb.south){$\ket{0}_{3i+1'}$};
			\node[rectangle,draw,anchor=west](outState1ab)	at (stateb.east)			{$\ket{\psi_{3i,1}}_{	3i+1,	3i,1}$};
			\node[rectangle,draw,anchor=west](outState1bb)	at (outState1ab.east)	{$\ket{\psi_{3i+2,1}}_{	3i+1,	3i+2,1}$};
			\node[rectangle,draw,anchor=north](outState2ab)	at (outState1ab.south)	{$\ket{\psi_{3i,2}}_{	3i+1,	3i,2}$};
			\node[rectangle,draw,anchor=west](outState2bb)	at (outState2ab.east)	{$\ket{\psi_{3i+2,2}}_{	3i+1,	3i+2,2}$};
			\node[rectangle,draw,anchor=north east](statec) at (outState2ab.south west){$\ket{\psi}_{3i+2}$};
			\node[rectangle,draw,anchor=north](stateproxc) at (statec.south){$\ket{0}_{3i+2'}$}; 
			\node[rectangle,draw,anchor=west](outState1ac)	at (statec.east)			{$\ket{\psi_{3i,1}}_{	3i+2,	3i,1}$};
			\node[rectangle,draw,anchor=west](outState1bc)	at (outState1ac.east)	{$\ket{\psi_{3i+1,1}}_{	3i+2,	3i+1,1}$};
			\node[rectangle,draw,anchor=north](outState2ac)	at (outState1ac.south)	{$\ket{\psi_{3i,2}}_{	3i+2,	3i,2}$};
			\node[rectangle,draw,anchor=west](outState2bc)	at (outState2ac.east)	{$\ket{\psi_{3i+1,2}}_{	3i+2,	3i+1,2}$};
		\end{tikzpicture}
	\end{figure}
	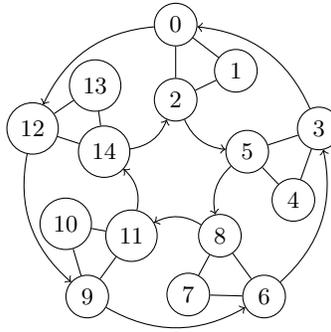
\begin{figure}[!ht]
		\centering
		\caption{15 node WQCA with node interacting edges}
		\label{fig:15WQCAgraph}
		\begin{tikzpicture}
			\graph[nodes={draw, circle,minimum size= 1em},edge={<-,bend left}, cycle, n=5, clockwise, radius=2cm, name=A]
			{
				1/"0", 2/"3", 3/"6", 4/"9", 5/"12"
			};
			\graph[nodes={draw, circle,minimum size= 1em}, n=5, clockwise, radius=1.6cm,phase=60,name = B]
			{
				1/"1", 2/"4", 3/"7", 4/"10", 5/"13"
			};
			\graph[nodes={draw, circle,minimum size= 1em},edge={->,bend right}, cycle, n=5, clockwise, radius=1cm, name = C]
			{
				1/"2", 2/"5", 3/"8", 4/"11", 5/"14"
			};
			\foreach \i [evaluate={\j=int(mod(\i+1+4,5)+1)}]
			in {1,2,3,4,5}{
				\draw (A \i) -- (B \i) -- (C \i) -- (A \i);
			}
		\end{tikzpicture}
	\end{figure}
	As each node's $\quoted{i,\square,1}$ register is initialized as a 2-W-State, and bell paired with $\quoted{i,\square,2}$ we can make certain assumptions about measurement outcomes.  By design if $\quoted{i,j,1},\quoted{j,i,2}$ is measured as $1,1$, then necessarily $\forall u \in \neigh{i}\setminus j$ and $\forall v \in \neigh{j}\setminus i$ $\quoted{i,u,1},\quoted{u,i,2}$ or $\quoted{j,v,1},\quoted{v,j,2}$ cannot be measured as $1,1$.  Thus using them as controls on a $\mathtt{SWAP}$ will act as a semaphore on potential interacting operations. For example, examining the behavior of the unitary on the block of $\quoted{0},\quoted{1},\quoted{2}$ we have:
	\begin{equation}
		\begin{aligned}
			I 
			+	\ket{1}_{0,1,1}\ket{1}_{1,0,2}	\left( \mathtt{SWAP}_{0,1}	-	I	\right)	\bra{1}_{0,1,1}\bra{1}_{1,0,2}\\
			+	\ket{1}_{0,2,1}\ket{1}_{2,0,2}	\left( \mathtt{SWAP}_{0,2}	-	I	\right)	\bra{1}_{0,2,1}\bra{1}_{2,0,2}\\
			+	\ket{1}_{1,2,1}\ket{1}_{2,1,2}	\left( \mathtt{SWAP}_{1,2}	-	I	\right)	\bra{1}_{1,2,1}\bra{1}_{2,1,2}
		\end{aligned}
	\end{equation}
	In a simplified Watrous QCA let us have a graph cycle $C_{15}$, assigning a  qubit to each vertex $i$ labeled $\quoted{i}$. After the first step of WQCA, we have block contents swapping left and right via the inner and outer cycles as pictured in figure \ref{fig:15WQCAgraph} with the application of:
	\begin{equation}
		\bigotimes_{i=0}^{4}{\mathtt{SWAP}_{3i,3i+3\mathbin{\%}15} \mathtt{SWAP}_{3i,3i-3\mathbin{\%}15}}
	\end{equation}
	Initializing the system with all but one qubit as $0$ and one as $1$, we have:
	\begin{equation}
		\ket{1} \foreach \x in {2,...,15}{\ket{0}}
	\end{equation}
	Watrous's permutation gives us the state:
	\begin{equation}
		\foreach \x in {1,...,12}{\ket{0}}
		\ket{1} \foreach \x in {14,...,15}{\ket{0}} 
		= \ket{0}^{\otimes 12} \ket{1} \ket{0}	\ket{0}
	\end{equation}
	including the auxiliary states we have:
	\begin{equation}
		\begin{aligned}
			\foreach \x in {1,...,12}{\ket{0}}
			\ket{1} \foreach \x in {14,...,15}{\ket{0}} \\
			= \ket{0}^{\otimes 12} \ket{1} \ket{0}	\ket{0}
			\otimes \foreach \x in {0,...,5}{\ket{\psi}_{\x,\square,\square}}\\
			\foreach \x in {6,...,11}{\ket{\psi}_{\x,\square,\square}}\\
			\resizebox{.9\linewidth}{!}{$
				\begin{aligned}
			\frac{1}{\sqrt{2}}\left(\ket{0}_{12,13,1}\ket{0}_{12,13,2}+\ket{1}_{12,13,1}\ket{1}_{12,13,2}\right)
			\frac{1}{\sqrt{2}}\left(\ket{0}_{12,14,1}\ket{0}_{12,14,2}+\ket{1}_{12,14,1}\ket{1}_{12,14,2}\right)\\
			\frac{1}{\sqrt{2}}\left(\ket{0}_{13,12,1}\ket{0}_{13,12,2}+\ket{1}_{13,12,1}\ket{1}_{13,12,2}\right)
			\frac{1}{\sqrt{2}}\left(\ket{0}_{13,14,1}\ket{0}_{13,14,2}+\ket{1}_{13,14,1}\ket{1}_{13,14,2}\right)\\
			\frac{1}{\sqrt{2}}\left(\ket{0}_{14,12,1}\ket{0}_{14,12,2}+\ket{1}_{14,12,1}\ket{1}_{14,12,2}\right)
			\frac{1}{\sqrt{2}}\left(\ket{0}_{14,13,1}\ket{0}_{14,13,2}+\ket{1}_{14,13,1}\ket{1}_{14,13,2}\right)
			\end{aligned}$}
		\end{aligned}
	\end{equation}
	As nodes $0,...,11$ are all 0, their transformations can be omitted as they are in a quiescent state.  Applying the transformation to the remaining block of qubits gives us:
	\foreach \a/\b in {0/1,1/0}\foreach \c/\d in {0/1,1/0}\foreach \x/\y in {0/1,1/0}{
		\[\vspace{-1em}\ifnum\numexpr(\a+\c+\x)>0 + \fi
		\ket{\a}_{12,13,1}	\ket{\c}_{13,12,2}
		\ket{\b}_{12,14,1}	\ket{\x}_{14,12,2}
		\ket{\d}_{13,14,1}	\ket{\y}_{14,13,2}
		\ifnum\numexpr(\a+\c)=2 \ket{0}_{12}\ket{1}_{13}\ket{0}_{14}	\else \ifnum\numexpr(\b+\x)=2	\ket{0}_{12}\ket{0}_{13}\ket{1}_{14}	\else\ifnum\numexpr(\d+\y)=2	\ket{1}_{12}\ket{0}_{13}\ket{0}_{14} \else	\ket{1}_{12}\ket{0}_{13}\ket{0}_{14}\fi	\fi	\fi\]}

	\vspace{1em}Taking the partial trace of the system's edge flags shows that the resultant system is in a superposition of the 1 value being present in node $12,13,$ or $14$ with probability:
	\begin{tabular}{|c|c|}
		\hline
		Node & Prob \\
		\hline
		12 & $\frac{1}{2}$\\
		\hline
		13 & $\frac{1}{4}$\\
		\hline
		14 & $\frac{1}{4}$\\
		\hline
	\end{tabular}
	
	\subsection{Example}
	We initialize the system in the state:
	
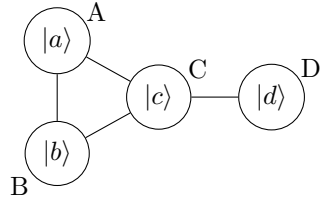
\[\begin{tikzpicture}[scale=.9,node distance={15mm},state/.style={draw, circle,minimum size=8mm}]
		\node[state](A) {$\ket{a}$};
		\node[xshift=3mm] at (A.60) {A};
		\node[state](B) [below of=A] {$\ket{b}$};
		\node[xshift=-5mm] at (B.270) {B};
		\node[state](C) at ($.5*($(A) + (B)$)+(1.5,0) $) {$\ket{c}$};
		\node[xshift=3mm] at (C.60) {C};
		\node[state](D) [right of=C] {$\ket{d}$};
		\node[xshift=3mm] at (D.60) {D};
		\draw[below] (A) -- (B) -- (C) -- (D);
		\draw[below] (A) -- (C);
	\end{tikzpicture}
	\]\[
	\begin{aligned}
		\ket{a}_A\ket{0}_{A,B,1}\ket{0}_{A,C,1}\ket{0}_{A,B,2}\ket{0}_{A,C,2}
		\ket{b}_B\ket{0}_{B,A,1}\ket{0}_{B,C,1}\ket{0}_{B,A,2}\ket{0}_{B,C,2}\\
		\ket{c}_C\ket{0}_{C,A,1}\ket{0}_{C,B,1}\ket{0}_{C,D,1}\ket{0}_{C,A,2}\ket{0}_{C,B,2}\ket{0}_{C,D,2}
		\ket{d}_D\ket{0}_{D,C,1}\ket{0}_{D,C,2}
	\end{aligned}
	\]
	where $\ket{a},\ket{b},\ket{c},\ket{d}$ are the quantum states of vertices $A,B,C,$ and $D$ respectively.  We assume here that these are all distinct states, but there is no requirement for them to be.
	Applying $U_W^{A}\circ U_W^{B}\circ U_W^{C}\circ U_W^{D}$ to the system gives
	\[
	\begin{aligned}
		\ket{a}_A
		\frac{1}{\sqrt{2}}\left(
		\ket{1}_{A,B,1}\ket{0}_{A,C,1}\ket{1}_{A,B,2}\ket{0}_{A,C,2}+
		\ket{0}_{A,B,1}\ket{1}_{A,C,1}\ket{0}_{A,B,2}\ket{1}_{A,C,2} \right)\\
		\ket{b}_B
		\frac{1}{\sqrt{2}}\left(
		\ket{1}_{B,A,1}\ket{0}_{B,C,1}\ket{1}_{B,A,2}\ket{0}_{B,C,2}+
		\ket{0}_{B,A,1}\ket{1}_{B,C,1}\ket{0}_{B,A,2}\ket{1}_{B,C,2}\right)\\
		\ket{c}_C
		\frac{1}{\sqrt{3}}\left(
		\begin{aligned}
			\ket{1}_{C,A,1}\ket{0}_{C,B,1}\ket{0}_{C,D,1}\ket{1}_{C,A,2}\ket{0}_{C,B,2}\ket{0}_{C,D,2}\\+
			\ket{0}_{C,A,1}\ket{1}_{C,B,1}\ket{0}_{C,D,1}\ket{0}_{C,A,2}\ket{1}_{C,B,2}\ket{0}_{C,D,2}\\+
			\ket{0}_{C,A,1}\ket{0}_{C,B,1}\ket{1}_{C,D,1}\ket{0}_{C,A,2}\ket{0}_{C,B,2}\ket{1}_{C,D,2}\end{aligned}\right)\\
		\ket{d}_D\ket{1}_{D,C,1}\ket{1}_{D,C,2}
	\end{aligned}
	\]
	which expands to a sum of states such as
	\[
	\begin{aligned}
		\ket{a}_A\ket{b}_B\ket{c}_C\ket{d}_D\ket{1}_{A,B,1}\ket{1}_{A,C,1}\ket{1}_{A,B,2}\ket{0}_{A,C,2}\ket{1}_{B,A,1}\ket{0}_{B,C,1}\ket{1}_{B,A,2}\\
		\otimes \ket{0}_{B,C,2}\ket{0}_{C,A,1}\ket{0}_{C,B,1}\ket{1}_{C,D,1}\ket{1}_{C,A,2}\ket{0}_{C,B,2}\ket{0}_{C,D,2}\ket{0}_{D,C,1}\ket{1}_{D,C,2}
	\end{aligned}\]
	which can be viewed as the state
	
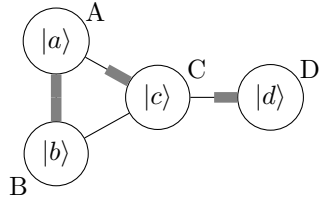
\[	\begin{tikzpicture}[scale=.9,node distance={15mm},state/.style={draw, circle,minimum size=8mm},highlight/.style={draw, gray, line width=4pt}]
		\node[state](A) {$\ket{a}$};
		\node[xshift=3mm] at (A.60) {A};
		\node[state](B) [below of=A] {$\ket{b}$};
		\node[xshift=-5mm] at (B.270) {B};
		\node[state](C) at ($.5*($(A) + (B)$)+(1.5,0) $) {$\ket{c}$};
		\node[xshift=3mm] at (C.60) {C};
		\node[state](D) [right of=C] {$\ket{d}$};
		\node[xshift=3mm] at (D.60) {D};
		\draw[below] (A) -- (B) -- (C) -- (D);
		\draw[below] (A) -- (C);
		\draw[highlight] (A) edge ($(A)!.5!(B)$) ;
		\draw[highlight] (B) edge ($(B)!.5!(A)$) ;
		\draw[highlight] (C) edge ($(C)!.5!(A)$) ;
		\draw[highlight] (D) edge ($(D)!.5!(C)$) ;
	\end{tikzpicture}\]
	where the edge $\left(A,B\right)$ is fully highlighted as $\ket{1}_{A,B,2}$ and $\ket{1}_{B,A,2}$.  Similarly half of $\left(A,B\right)$ is highlighted as $\ket{0}_{A,C,2}$ and $\ket{1}_{C,A,2}$.

	Applying $\unitaryComposition_{\left(u,v\right) \in E}{\text{CCSWAP}_{\left(u,v,1\right),\left(u,v,2\right),\left(u\right),\left(v\right)}}$ to this fragment of the system results in:
	\[
	\begin{aligned}
		\ket{b}_A\ket{a}_B\ket{c}_C\ket{d}_D\ket{1}_{A,B,1}\ket{1}_{A,C,1}\ket{1}_{A,B,2}\ket{0}_{A,C,2}\ket{1}_{B,A,1}\ket{0}_{B,C,1}\ket{1}_{B,A,2}\\
		\otimes \ket{0}_{B,C,2}\ket{0}_{C,A,1}\ket{0}_{C,B,1}\ket{1}_{C,D,1}\ket{1}_{C,A,2}\ket{0}_{C,B,2}\ket{0}_{C,D,2}\ket{0}_{D,C,1}\ket{1}_{D,C,2}
	\end{aligned}\]
	or
	
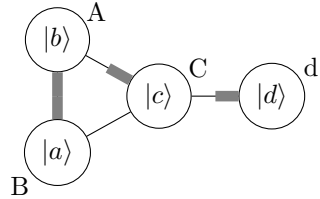
\[	\begin{tikzpicture}[scale=.9,node distance={15mm},state/.style={draw, circle,minimum size=8mm},highlight/.style={draw, gray, line width=4pt}]
		\node[state](A) {$\ket{b}$};
		\node[xshift=3mm] at (A.60) {A};
		\node[state](B) [below of=A] {$\ket{a}$};
		\node[xshift=-5mm] at (B.270) {B};
		\node[state](C) at ($.5*($(A) + (B)$)+(1.5,0) $) {$\ket{c}$};
		\node[xshift=3mm] at (C.60) {C};
		\node[state](D) [right of=C] {$\ket{d}$};
		\node[xshift=3mm] at (D.60) {d};
		\draw[below] (A) -- (B) -- (C) -- (D);
		\draw[below] (A) -- (C);
		\draw[highlight] (A) edge ($(A)!.5!(B)$) ;
		\draw[highlight] (B) edge ($(B)!.5!(A)$) ;
		\draw[highlight] (C) edge ($(C)!.5!(A)$) ;
		\draw[highlight] (D) edge ($(D)!.5!(C)$) ;
	\end{tikzpicture}\]
	The full list of such configurations then is:
	
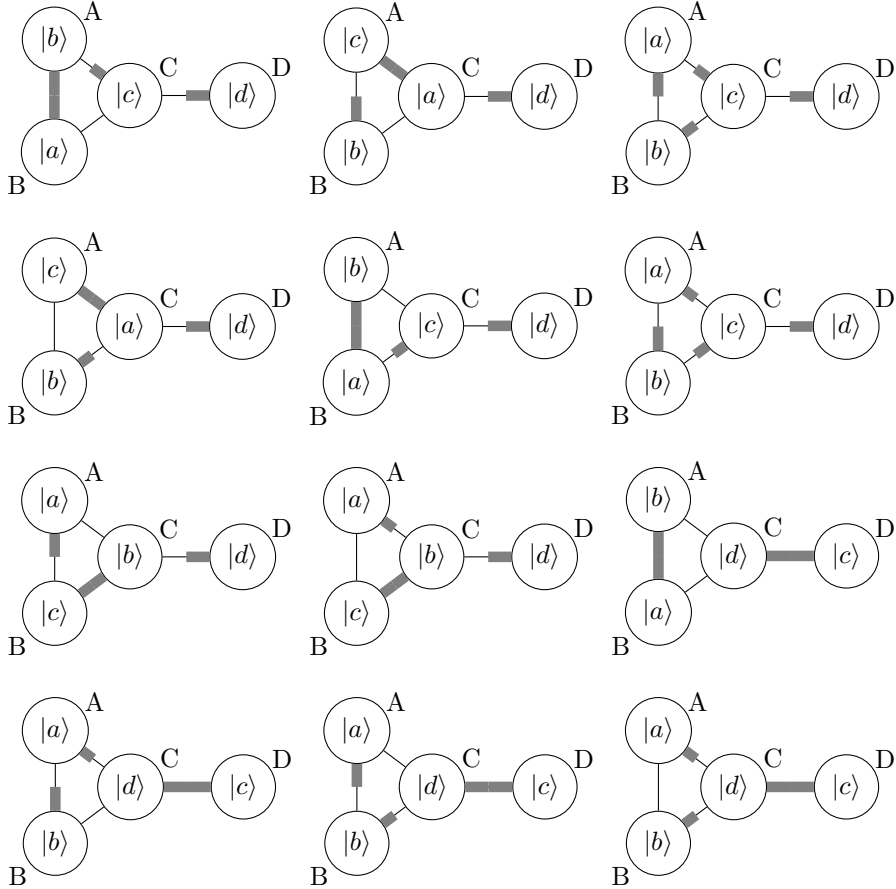
\[\begin{aligned}
		\begin{tikzpicture}[scale=.20,node distance={15mm},state/.style={draw, circle,minimum size=8mm},highlight/.style={draw, gray, line width=4pt}]
			\node[state](A) {$\ket{b}$};
			\node[state](B) [below of=A] {$\ket{a}$};
			\node[state](C) at ($.5*($(A) + (B)$)+(5,0) $) {$\ket{c}$};
			\node[state](D) [right of=C] {$\ket{d}$};
			\draw[below] (A) -- (B) -- (C) -- (D);
			\draw[below] (A) -- (C);
			\node[xshift=3mm] at (A.60) {A};
			\node[xshift=-5mm] at (B.270) {B};
			\node[xshift=3mm] at (C.60) {C};
			\node[xshift=3mm] at (D.60) {D};
			\draw[highlight] (A) edge ($(A)!.5!(B)$) ;
			\draw[highlight] (B) edge ($(B)!.5!(A)$) ;
			\draw[highlight] (C) edge ($(C)!.5!(A)$) ;
			\draw[highlight] (D) edge ($(D)!.5!(C)$) ;
		\end{tikzpicture}
		\begin{tikzpicture}[scale=.20,node distance={15mm},state/.style={draw, circle,minimum size=8mm},highlight/.style={draw, gray, line width=4pt}]
			\node[state](A) {$\ket{c}$};
			\node[state](B) [below of=A] {$\ket{b}$};
			\node[state](C) at ($.5*($(A) + (B)$)+(5,0) $) {$\ket{a}$};
			\node[state](D) [right of=C] {$\ket{d}$};
			\draw[below] (A) -- (B) -- (C) -- (D);
			\draw[below] (A) -- (C);
			\node[xshift=3mm] at (A.60) {A};
			\node[xshift=-5mm] at (B.270) {B};
			\node[xshift=3mm] at (C.60) {C};
			\node[xshift=3mm] at (D.60) {D};
			\draw[highlight] (A) edge ($(A)!.5!(C)$) ;
			\draw[highlight] (B) edge ($(B)!.5!(A)$) ;
			\draw[highlight] (C) edge ($(C)!.5!(A)$) ;
			\draw[highlight] (D) edge ($(D)!.5!(C)$) ;
		\end{tikzpicture}
		\begin{tikzpicture}[scale=.20,node distance={15mm},state/.style={draw, circle,minimum size=8mm},highlight/.style={draw, gray, line width=4pt}]
			\node[state](A) {$\ket{a}$};
			\node[state](B) [below of=A] {$\ket{b}$};
			\node[state](C) at ($.5*($(A) + (B)$)+(5,0) $) {$\ket{c}$};
			\node[state](D) [right of=C] {$\ket{d}$};
			\draw[below] (A) -- (B) -- (C) -- (D);
			\draw[below] (A) -- (C);
			\node[xshift=3mm] at (A.60) {A};
			\node[xshift=-5mm] at (B.270) {B};
			\node[xshift=3mm] at (C.60) {C};
			\node[xshift=3mm] at (D.60) {D};
			\draw[highlight] (A) edge ($(A)!.5!(B)$) ;
			\draw[highlight] (B) edge ($(B)!.5!(C)$) ;
			\draw[highlight] (C) edge ($(C)!.5!(A)$) ;
			\draw[highlight] (D) edge ($(D)!.5!(C)$) ;
		\end{tikzpicture}
		\\\begin{tikzpicture}[scale=.20,node distance={15mm},state/.style={draw, circle,minimum size=8mm},highlight/.style={draw, gray, line width=4pt}]
			\node[state](A) {$\ket{c}$};
			\node[state](B) [below of=A] {$\ket{b}$};
			\node[state](C) at ($.5*($(A) + (B)$)+(5,0) $) {$\ket{a}$};
			\node[state](D) [right of=C] {$\ket{d}$};
			\draw[below] (A) -- (B) -- (C) -- (D);
			\draw[below] (A) -- (C);
			\node[xshift=3mm] at (A.60) {A};
			\node[xshift=-5mm] at (B.270) {B};
			\node[xshift=3mm] at (C.60) {C};
			\node[xshift=3mm] at (D.60) {D};
			\draw[highlight] (A) edge ($(A)!.5!(C)$) ;
			\draw[highlight] (B) edge ($(B)!.5!(C)$) ;
			\draw[highlight] (C) edge ($(C)!.5!(A)$) ;
			\draw[highlight] (D) edge ($(D)!.5!(C)$) ;
		\end{tikzpicture}
		\begin{tikzpicture}[scale=.20,node distance={15mm},state/.style={draw, circle,minimum size=8mm},highlight/.style={draw, gray, line width=4pt}]
			\node[state](A) {$\ket{b}$};
			\node[state](B) [below of=A] {$\ket{a}$};
			\node[state](C) at ($.5*($(A) + (B)$)+(5,0) $) {$\ket{c}$};
			\node[state](D) [right of=C] {$\ket{d}$};
			\draw[below] (A) -- (B) -- (C) -- (D);
			\draw[below] (A) -- (C);
			\node[xshift=3mm] at (A.60) {A};
			\node[xshift=-5mm] at (B.270) {B};
			\node[xshift=3mm] at (C.60) {C};
			\node[xshift=3mm] at (D.60) {D};
			\draw[highlight] (A) edge ($(A)!.5!(B)$) ;
			\draw[highlight] (B) edge ($(B)!.5!(A)$) ;
			\draw[highlight] (C) edge ($(C)!.5!(B)$) ;
			\draw[highlight] (D) edge ($(D)!.5!(C)$) ;
		\end{tikzpicture}
		\begin{tikzpicture}[scale=.20,node distance={15mm},state/.style={draw, circle,minimum size=8mm},highlight/.style={draw, gray, line width=4pt}]
			\node[state](A) {$\ket{a}$};
			\node[state](B) [below of=A] {$\ket{b}$};
			\node[state](C) at ($.5*($(A) + (B)$)+(5,0) $) {$\ket{c}$};
			\node[state](D) [right of=C] {$\ket{d}$};
			\draw[below] (A) -- (B) -- (C) -- (D);
			\draw[below] (A) -- (C);
			\node[xshift=3mm] at (A.60) {A};
			\node[xshift=-5mm] at (B.270) {B};
			\node[xshift=3mm] at (C.60) {C};
			\node[xshift=3mm] at (D.60) {D};
			\draw[highlight] (A) edge ($(A)!.5!(C)$) ;
			\draw[highlight] (B) edge ($(B)!.5!(A)$) ;
			\draw[highlight] (C) edge ($(C)!.5!(B)$) ;
			\draw[highlight] (D) edge ($(D)!.5!(C)$) ;
		\end{tikzpicture}
		\\\begin{tikzpicture}[scale=.20,node distance={15mm},state/.style={draw, circle,minimum size=8mm},highlight/.style={draw, gray, line width=4pt}]
			\node[state](A) {$\ket{a}$};
			\node[state](B) [below of=A] {$\ket{c}$};
			\node[state](C) at ($.5*($(A) + (B)$)+(5,0) $) {$\ket{b}$};
			\node[state](D) [right of=C] {$\ket{d}$};
			\draw[below] (A) -- (B) -- (C) -- (D);
			\draw[below] (A) -- (C);
			\node[xshift=3mm] at (A.60) {A};
			\node[xshift=-5mm] at (B.270) {B};
			\node[xshift=3mm] at (C.60) {C};
			\node[xshift=3mm] at (D.60) {D};
			\draw[highlight] (A) edge ($(A)!.5!(B)$) ;
			\draw[highlight] (B) edge ($(B)!.5!(C)$) ;
			\draw[highlight] (C) edge ($(C)!.5!(B)$) ;
			\draw[highlight] (D) edge ($(D)!.5!(C)$) ;
		\end{tikzpicture}
		\begin{tikzpicture}[scale=.20,node distance={15mm},state/.style={draw, circle,minimum size=8mm},highlight/.style={draw, gray, line width=4pt}]
			\node[state](A) {$\ket{a}$};
			\node[state](B) [below of=A] {$\ket{c}$};
			\node[state](C) at ($.5*($(A) + (B)$)+(5,0) $) {$\ket{b}$};
			\node[state](D) [right of=C] {$\ket{d}$};
			\draw[below] (A) -- (B) -- (C) -- (D);
			\draw[below] (A) -- (C);
			\node[xshift=3mm] at (A.60) {A};
			\node[xshift=-5mm] at (B.270) {B};
			\node[xshift=3mm] at (C.60) {C};
			\node[xshift=3mm] at (D.60) {D};
			\draw[highlight] (A) edge ($(A)!.5!(C)$) ;
			\draw[highlight] (B) edge ($(B)!.5!(C)$) ;
			\draw[highlight] (C) edge ($(C)!.5!(B)$) ;
			\draw[highlight] (D) edge ($(D)!.5!(C)$) ;
		\end{tikzpicture}
		\begin{tikzpicture}[scale=.20,node distance={15mm},state/.style={draw, circle,minimum size=8mm},highlight/.style={draw, gray, line width=4pt}]
			\node[state](A) {$\ket{b}$};
			\node[state](B) [below of=A] {$\ket{a}$};
			\node[state](C) at ($.5*($(A) + (B)$)+(5,0) $) {$\ket{d}$};
			\node[state](D) [right of=C] {$\ket{c}$};
			\draw[below] (A) -- (B) -- (C) -- (D);
			\draw[below] (A) -- (C);
			\node[xshift=3mm] at (A.60) {A};
			\node[xshift=-5mm] at (B.270) {B};
			\node[xshift=3mm] at (C.60) {C};
			\node[xshift=3mm] at (D.60) {D};
			\draw[highlight] (A) edge ($(A)!.5!(B)$) ;
			\draw[highlight] (B) edge ($(B)!.5!(A)$) ;
			\draw[highlight] (C) edge ($(C)!.5!(D)$) ;
			\draw[highlight] (D) edge ($(D)!.5!(C)$) ;
		\end{tikzpicture}
		\\\begin{tikzpicture}[scale=.20,node distance={15mm},state/.style={draw, circle,minimum size=8mm},highlight/.style={draw, gray, line width=4pt}]
			\node[state](A) {$\ket{a}$};
			\node[state](B) [below of=A] {$\ket{b}$};
			\node[state](C) at ($.5*($(A) + (B)$)+(5,0) $) {$\ket{d}$};
			\node[state](D) [right of=C] {$\ket{c}$};
			\draw[below] (A) -- (B) -- (C) -- (D);
			\draw[below] (A) -- (C);
			\node[xshift=3mm] at (A.60) {A};
			\node[xshift=-5mm] at (B.270) {B};
			\node[xshift=3mm] at (C.60) {C};
			\node[xshift=3mm] at (D.60) {D};
			\draw[highlight] (A) edge ($(A)!.5!(C)$) ;
			\draw[highlight] (B) edge ($(B)!.5!(A)$) ;
			\draw[highlight] (C) edge ($(C)!.5!(D)$) ;
			\draw[highlight] (D) edge ($(D)!.5!(C)$) ;
		\end{tikzpicture}
		\begin{tikzpicture}[scale=.20,node distance={15mm},state/.style={draw, circle,minimum size=8mm},highlight/.style={draw, gray, line width=4pt}]
			\node[state](A) {$\ket{a}$};
			\node[state](B) [below of=A] {$\ket{b}$};
			\node[state](C) at ($.5*($(A) + (B)$)+(5,0) $) {$\ket{d}$};
			\node[state](D) [right of=C] {$\ket{c}$};
			\draw[below] (A) -- (B) -- (C) -- (D);
			\draw[below] (A) -- (C);
			\node[xshift=3mm] at (A.60) {A};
			\node[xshift=-5mm] at (B.270) {B};
			\node[xshift=3mm] at (C.60) {C};
			\node[xshift=3mm] at (D.60) {D};
			\draw[highlight] (A) edge ($(A)!.5!(B)$) ;
			\draw[highlight] (B) edge ($(B)!.5!(C)$) ;
			\draw[highlight] (C) edge ($(C)!.5!(D)$) ;
			\draw[highlight] (D) edge ($(D)!.5!(C)$) ;
		\end{tikzpicture}
		\begin{tikzpicture}[scale=.20,node distance={15mm},state/.style={draw, circle,minimum size=8mm},highlight/.style={draw, gray, line width=4pt}]
			\node[state](A) {$\ket{a}$};
			\node[state](B) [below of=A] {$\ket{b}$};
			\node[state](C) at ($.5*($(A) + (B)$)+(5,0) $) {$\ket{d}$};
			\node[state](D) [right of=C] {$\ket{c}$};
			\draw[below] (A) -- (B) -- (C) -- (D);
			\draw[below] (A) -- (C);
			\node[xshift=3mm] at (A.60) {A};
			\node[xshift=-5mm] at (B.270) {B};
			\node[xshift=3mm] at (C.60) {C};
			\node[xshift=3mm] at (D.60) {D};
			\draw[highlight] (A) edge ($(A)!.5!(C)$) ;
			\draw[highlight] (B) edge ($(B)!.5!(C)$) ;
			\draw[highlight] (C) edge ($(C)!.5!(D)$) ;
			\draw[highlight] (D) edge ($(D)!.5!(C)$) ;
		\end{tikzpicture}
	\end{aligned}\]
	
	This enumeration of all potential configurations then illustrate that result of our process, a state where the values of vertices is superpositioned over all potential walker behaviors.  By design, no two adjacent edges are fully highlighted, and the probability of an edge being selected is independent of the selection of all non-adjacent edges.  We then have the superposition of these states as one round of diffusion of the states $\ket{a}$ through $\ket{d}$ which can be seen as one iteration of diffusion of vertex states. The resultant diffused state of the initial system then is $\frac{1}{\sqrt{12}}$ the sum of these 12 graphs in their Dirac notational form.
	
	To illustrate the results of this diffusion for example, when measuring $\quoted{A}$, we would take the partial trace of the state giving us
	\[
	\frac{7}{12}\ket{a}_{A}\bra{a}_{A}+
	\frac{3}{12}\ket{b}_{A}\bra{b}_{A}+
	\frac{2}{12}\ket{c}_{A}\bra{c}_{A}
	\]
	which corresponds to the probability distribution:
	\[\begin{tabular}{|c c|}
		\hline
		M(A)& \text{Prob}\\
		\hline
		a& $\frac{7}{12}$\\
		b& $\frac{3}{12}$\\
		c& $\frac{2}{12}$\\
		d& $\frac{0}{12}$\\
		\hline
	\end{tabular}\]
	
	\subsection{Directed Diffusion}
	By replacing CCSWAP with another doubly controlled transformation which only swaps particular states we can introduce much more complicated behaviors.  For example let us have two 3 level qudits $\ket{i}\ket{j}$ and a unitary that swap $\ket{0}$ with any state and will not swap $\ket{1}$ or $\ket{2}$.
	Then 
	\[
	\begin{bmatrix}
		1&  0&  0&  0&  0&  0&  0&  0& 0 \\
		0&  0&  0&  1&  0&  0&  0&  0& 0 \\
		0&  0&  0&  0&  0&  0&  1&  0& 0 \\
		0&  1&  0&  0&  0&  0&  0&  0& 0 \\
		0&  0&  0&  0&  1&  0&  0&  0& 0 \\
		0&  0&  0&  0&  0&  1&  0&  0& 0 \\
		0&  0&  1&  0&  0&  0&  0&  0& 0 \\
		0&  0&  0&  0&  0&  0&  0&  1& 0 \\
		0&  0&  0&  0&  0&  0&  0&  0& 1   
	\end{bmatrix}\]
	acts as a way for walkers presence in a vertex to effect how other walkers spread to that vertex.
	\section{Conclusion and Future work}
	In this paper we've defined an alternative form of implementing a multi part quantum walker on a graph.  While our implementation clearly enlarges the necessary quantum computational space needed to model and iterate a quantum walk it presents several distinctive advantages over previous approaches.  
	As the coins driving the walkers are distributed across all vertices according to the total degree of a vertex and all steps of the vertex's transformation are driven by the states of it and its neighbors, the process does not require full knowledge of the graph to implement.
	The potential substation of more conditional doubly controlled swaps allows us to further influence the evolution of the system while remaining largely ignorant of its exact state.
	
	As our alternative implementation of quantum walkers is significantly more resource demanding, we will first work toward minimizing the needed qubits to simulate the system's behavior.  Once we are able to simulate non-trivial graphs, we will begin to study the behavior of both undirected and directed diffusion on walkers.  Specifically we will construct graphs to evaluate the degree to which directed diffusion can be leveraged in optimizing the transmission of a walker from specified sources to specified endpoints.  By applying directed diffusion and using periodic partial measurements to condense the resultant diffused states could lead to a robust blind gradient of states traversing the graph.  In cluster state computation\cite{nielsen_fault-tolerant_2004} the combination of measurements with unitary operations drives quantum computations.  Combining cluster states techniques with directed diffusion will be investigated in future work.
	\bibliography{references}

\begin{thebibliography}{10}

\bibitem{brun_quantum_2003}
Todd~A. Brun, Hilary~A. Carteret, and Andris Ambainis.
\newblock Quantum walks driven by many coins.
\newblock {\em Physical Review A - Atomic, Molecular, and Optical Physics}, 67(5):17, 2003.
\newblock arXiv: quant-ph/0210161.

\bibitem{du_three_2000}
W~Dü, G~Vidal, and J~I Cirac.
\newblock Three qubits can be entangled in two inequivalent ways.
\newblock {\em Physical Review A}, 62:1--12, 2000.

\bibitem{hassanpour_bidirectional_2015}
Shima Hassanpour and Monireh Houshmand.
\newblock Bidirectional quantum teleportation and secure direct communication via entanglement swapping.
\newblock IEEE, 2015.

\bibitem{keating_localization_2007}
J.~P. Keating, N.~Linden, J.~C.F. Matthews, and A.~Winter.
\newblock Localization and its consequences for quantum walk algorithms and quantum communication.
\newblock {\em Physical Review A - Atomic, Molecular, and Optical Physics}, 76(1):6--9, 2007.
\newblock arXiv: quant-ph/0606205.

\bibitem{kempe_quantum_2003}
Julia Kempe.
\newblock Quantum random walks: {An} introductory overview.
\newblock {\em Contemporary Physics}, 44(4):307--327, 2003.
\newblock arXiv: quant-ph/0303081.

\bibitem{nielsen_quantum_2011}
Michael~a. Nielsen and Isaac~L. Chuang.
\newblock {\em Quantum {Computation} and {Quantum} {Information}: 10th {Anniversary} {Edition}}.
\newblock Cambridge University Press, 2011.
\newblock arXiv: 1011.1669v3 Publication Title: Cambridge University Press ISSN: 00029505.

\bibitem{nielsen_fault-tolerant_2004}
Michael~A. Nielsen and Christopher~M. Dawson.
\newblock Fault-tolerant quantum computation with cluster states.
\newblock {\em Physical Review A - Atomic, Molecular, and Optical Physics}, 71(4):1--31, May 2004.
\newblock arXiv: quant-ph/0405134.

\bibitem{rohde_multi-walker_2011}
Peter~P. Rohde, Andreas Schreiber, Martin Štefaňák, Igor Jex, and Christine Silberhorn.
\newblock Multi-walker discrete time quantum walks on arbitrary graphs, their properties and their photonic implementation.
\newblock {\em New Journal of Physics}, 13:1--9, 2011.
\newblock arXiv: 1006.5556.

\bibitem{van_meter_quantum_2014}
Rodney Van~Meter.
\newblock {\em Quantum {Networking}}.
\newblock John Wiley \& Sons, Ltd, Chichester, UK, April 2014.
\newblock Publication Title: Quantum Networking Issue: 1.

\bibitem{watrous_one-dimensional_1995}
John Watrous.
\newblock On one-dimensional quantum cellular automata.
\newblock In {\em Proceedings of {IEEE} 36th {Annual} {Foundations} of {Computer} {Science}}, pages 528--537. IEEE Comput. Soc. Press, 1995.

\end{thebibliography}
	\bibliographystyle{plain}
\end{document}